\documentclass[copyright,creativecommons]{eptcs}

\providecommand{\event}{DCM 2010}

\usepackage{breakurl}

\usepackage{amsmath,amsthm,amssymb}

\theoremstyle{plain}
\newtheorem{theorem}{Theorem}[section]

\newtheorem{lemma}[theorem]{Lemma}

\theoremstyle{definition}
\newtheorem{convention}{Convention}
\newtheorem{definition}[theorem]{Definition}
\newtheorem{example}[theorem]{Example}

\newtheorem{remark}[theorem]{Remark}

\newlength{\neweqnberlength}
\newcommand{\neweqnber}[1]{%
\refstepcounter{equation}%
\settoheight{\neweqnberlength}{\begin{minipage}{0pt}%
\[\]%
\end{minipage}}%
\parbox[c][0pt]{0pt}{\begin{minipage}{0pt}%
\[\label{#1}\]%
\end{minipage}}{\eqref{#1}}%
}

\newcommand{\define}[1]{{\it #1}}
\newcommand{\s}[1]{\{{#1}\}}

\newcommand{\bor}{\;|\;}
\newcommand{\isdefby}{\mathbin{~::=~}}

\newcommand{\duploli}{\rightarrow}
\newcommand{\putype}{\top}
\newcommand{\prodtype}{\times}
\newcommand{\entail}{\vdash}

\newcommand{\canon}[1]{{[\;{#1}\;]}}
\newcommand{\cocanon}[1]{{\{\;{#1}\;\}}}
\newcommand{\prodterm}[1]{{\langle\;{#1}\;\rangle}}
\newcommand{\lproj}[1]{{{\pi_1}{#1}}}
\newcommand{\rproj}[1]{{{\pi_2}{#1}}}
\newcommand{\lprojp}[1]{{\pi_1({#1})}}
\newcommand{\rprojp}[1]{{\pi_2({#1})}}
\newcommand{\puterm}{\ast}
\newcommand{\zeroterm}{{\bf 0}}

\newcommand{\natzero}{{\bar{0}}}
\newcommand{\succterm}{{\it succ}}
\newcommand{\predterm}{{\it pred}}
\newcommand{\iszero}{{\it iszero}}
\newcommand{\trueterm}{{\it t\!t}}
\newcommand{\falseterm}{{\it ff}}

\newcommand{\ifterm}{{\it if}}
\newcommand{\thenterm}{{\it then}}
\newcommand{\elseterm}{{\it else}}
\newcommand{\ifthenelse}[3]{%
{\ifterm\;{#1}\;\thenterm\;{#2}\;\elseterm\;{#3}}}

\newcommand{\monad}{M}

\newcommand{\bit}{{\it bit}}
\newcommand{\nat}{{\it int}}

\newcommand{\alg}[1]{\mathcal{#1}}

\newcommand{\eqac}{\mathbin{\simeq_{AC}}}

\newcommand{\eqax}{\mathbin{\simeq_{\it ax}}}

\newcommand{\denot}[1]{{[\![\,{#1}\,]\!]}}

\newcommand{\cat}[1]{\mathcal{#1}}
\newcommand{\kcat}[2]{\cat{#1}_{#2}}
\newcommand{\prodfunc}{\times}
\newcommand{\punit}{{\bf 1}}
\newcommand{\internalhom}{\Rightarrow}

\newcommand{\unit}{\eta}
\newcommand{\mult}{\mu}

\newcommand{\strength}{t}

\newcommand{\linealcat}{\cat{C}_{l}}
\newcommand{\wlinealcat}{\cat{C}^w_{l}}
\newcommand{\Set}{{\bf Set}}

\newcommand{\ket}[1]{{|{#1}\rangle}}

\title{Semantics of a Typed Algebraic Lambda-Calculus}

\author{Beno\^{\i}t Valiron
\institute{Laboratoire d'Informatique de Grenoble\\
Universit\'e Joseph Fourier\\Grenoble\\ France\\
\email{benoit.valiron@monoidal.net}}}

\def\titlerunning{Semantics of a Typed Algebraic Lambda-Calculus}
\def\authorrunning{Beno\^{i}t Valiron}

\begin{document}

\maketitle

\begin{abstract}
  Algebraic lambda-calculi have been studied in various ways, but
  their semantics remain mostly untouched. In this paper we propose a
  semantic analysis of a general simply-typed lambda-calculus endowed
  with a structure of vector space. We sketch the relation with two
  established vectorial lambda-calculi. Then we study the problems
  arising from the addition of a fixed point combinator and how to
  modify the equational theory to solve them. We sketch an algebraic
  vectorial PCF and its possible denotational interpretations.
  
\end{abstract}

\section{Introduction}

Notions of lambda-calculus with vectorial structures have at least
three distinct origins.
A first line of
work~\cite{breazu91polymorphic,blanqui99calculus,barbanera93combining},
from which the term ``algebraic lambda-calculus'' comes from, focuses
on general algebraic rewrite systems and studies the conditions needed
for obtaining properties such as confluence or strong normalization.
The second one is the calculus of Vaux \cite{vaux08algebraic},
building up upon the work of Ehrhard and Regnier
\cite{ehrhard03differential}. The goal here is to capture a notion of
differentiation within lambda-calculus.
Finally, algebraic lambda-calculus also arises in the work of Arrighi
and Dowek \cite{arrighi08linear} where they define a lambda-calculus
oriented towards quantum computation, in the style of Van
Tonder~\cite{tonder04lambda}.

Both \cite{arrighi08linear} and \cite{vaux08algebraic} are concerned
with a lambda-calculus endowed with a structure of vector space. They
both
acknowledge
the fact that for an untyped lambda-calculus, a naive rewrite system
renders the language inconsistent, as any term can be made equal to
the zero of the vectorial space of terms. However, coming from
different backgrounds, they provide different solutions to the
problem. In \cite{arrighi08linear}, the rewriting system is
restrained in order to avoid unwanted equalities of terms. In
\cite{vaux08algebraic}, the rewriting system is untouched, but the
scalars over which the vectorial structure is built are made into a
semiring with particular properties, making the system consistent.
Finally, \cite{sts09} shows that a type
system enforcing strong normalization is also a mean of solving the
problem.

In this paper, we turn to the question of a semantics for a
lambda-calculus endowed with a structure of vector space (or more
generally, a structure of module). Starting
with an untyped lambda-calculus and a naive rewrite system, we recall
where inconsistencies occur. Then we construct a simply-typed version
of the untyped language together with an equational description. In
this restricted setting, the rewrite system is sound, and we describe
a denotational semantics using a computational model a la
Moggi~\cite{moggi91notions}. We also show how one can relate this
language to the one described in \cite{arrighi08linear} and
\cite{vaux08algebraic}.  We then re-read the problems that occurred in
the untyped world, and find a simple solution for making the system
sound again in the presence of diverging terms, finding an agreement
with the solution in \cite{vaux08algebraic}. The solution in this
paper goes however a step further, proposing a denotational framework
for the calculus.

\subsection{An untyped calculus}
\label{def:untypedlang}

Consider a ring $(\alg{A},+,0,\times,1)$. Elements of $\alg{A}$ are
called \define{scalars}. We define a call-by-value language as
follows.
\begin{align*}
  s,t &\isdefby x\bor \lambda x.s\bor st
  \bor s+t\bor \alpha\cdot s\bor\zeroterm
  \bor\canon{s}\bor\cocanon{s},
  \\
  u,v &\isdefby x\bor \lambda x.u\bor uv \bor\canon{s},
\end{align*}
where $\alpha$ ranges over $\alg{A}$, and where $x$ ranges over a
fixed set of variables. Terms of the form $s,t$ are called
\define{computations} and terms of the form $u,v$ are called
\define{values}. We define variable substitution as usual and consider
terms up to $\alpha$-equivalence. The meanings of the unusual terms are
explained in the next section.

\subsection{A naive reduction system}
\label{sec:naivered}

A very naive reduction is to make the set of terms into a module over
a ring $\alg{A}$, with the term $\zeroterm$ as unit of the addition.
More precisely, a term $s$ reduces to a term $t$, written $s\to t$, if
there exist terms $s'$ and $t'$ respectively equivalent modulo
congruence, associativity and commutativity of $+$ to $s$ and $t$ such
that the relation $s'\to t'$ is derived from the rules of
Table~\ref{tab:betared}. Although we do not describe formally the
system here (a complete development is done in
Section~\ref{sec:opsemst}), the reduction should be straightforward
enough for the remainder of the discussion.

In particular, the addition is commutative and associative, the terms
$t-t$ and $0\cdot t$ equate the term $\zeroterm$.  All term
constructs are linear with respect to addition and scalar
multiplication except $\canon{{}_-}$, 
which ``lifts'' a computation
into a value. One can unlift it using $\cocanon{{}_-}$, and
retrieve the computation. Finally, the system is call-by-value: the
beta-reduction $(\lambda x.s)v$ reduces to $s[x\leftarrow v]$ only if
$v$ is a value.

\begin{table*}[tbh]
\begin{center}
\begin{tabular}{c}
\hline
\hline
\mbox{\begin{minipage}{4.5in}
\def\mynl{\\[1.3ex]}
\begin{center}
\vspace{1ex}
Group $E$
\mynl
\mbox{$\begin{array}{rcl@{\qquad}rcl@{\qquad}rcl@{\qquad}rcl}
    \alpha\cdot \zeroterm &\to& \zeroterm
    & \zeroterm+s&\to& s
    & \alpha\cdot(\beta\cdot s)&\to&
    (\alpha\beta)\cdot s
    \\
    \hspace{-5ex}(*)\quad0\cdot s&\to&\zeroterm
    & 1\cdot s&\to& s
    & \alpha\cdot(s+t)&\to& \alpha\cdot s+\alpha\cdot t
  \end{array}$}
\mynl
Group $F$
\mynl
\mbox{$\begin{array}{r@{}lcr@{}lcr}
    \alpha\cdot{}&s &+& \beta\cdot{}&s&\to &(\alpha+\beta)\cdot s
    \\ 
    \alpha\cdot{}&s &+& &s &\to& (\alpha+1)\cdot s
    \\ 
    &s &+ &&s &\to& (1+1)\cdot s
  \end{array}$}
\mynl
Group $A$
\mynl
\mbox{$\begin{array}{@{}
      r@{~}c@{~}l
      @{~\quad}
      r@{~}c@{~}l
      @{~\quad}
      r@{~}c@{~}l
      @{~\quad}
      r@{~}c@{~}l
      @{}}
    (s+t)r &\to& sr + tr
    & (\alpha\cdot s)r  &\to& \alpha\cdot(sr)
    & \zeroterm r &\to& \zeroterm
    \\
    r(s+t) &\to& rs + rt
    & r(\alpha\cdot s) &\to& \alpha\cdot(rs)
    & r\zeroterm &\to& \zeroterm
    \\
    \lambda x.(s+t) &\to&  \lambda x.s + \lambda x.t
    & \lambda x.(\alpha\cdot s)&\to&  \alpha\cdot \lambda x.s
    & \lambda x.\zeroterm &\to& \zeroterm
    \\
    \cocanon{s+t} &\to&  \cocanon{s}+\cocanon{t}
    & \cocanon{\alpha\cdot s} &\to&  \alpha\cdot \cocanon{s}
    & \cocanon{\zeroterm} &\to& \zeroterm
  \end{array}$}
\mynl
Group $B$
\mynl
\mbox{$\begin{array}{c@{\qquad\qquad}c}
    (\lambda x.s)v\to s[x\leftarrow v]
    &
    \cocanon{\canon{s}} \to s
  \end{array}$}
\end{center}
\end{minipage}}
\\
\hline
\end{tabular}
\end{center}
\caption{Reduction system $L$.}
\label{tab:betared}
\end{table*}

For example, the term $(\lambda fx.(fx)x)(y+z)$ reduces to
$\lambda f.(fy)y+\lambda f.(fz)z$. On the contrary, the
computation $(\lambda xf.(f\cocanon{x})\cocanon{x})\canon{y+z}$
reduces to the sum of terms
$
\lambda f.(fy)y+\lambda f.(fz)y+\lambda f.(fy)z+\lambda f.(fz)z.
$

It is possible to build the same term constructs as with the regular
untyped lambda-calculus~\cite{barendregt84lambda}. For example, the
product $\prodterm{s,t}$ of two terms $s$ and $t$ can be encoded as
$\lambda f.(fs)t$, the first projection $\lprojp{s}$ of a pair $s$ as
the term $s\,(\lambda xy.x)$ and the second projection $\rprojp{s}$ as
$s\,(\lambda xy.y)$. Note that, since all usual lambda-term constructs
are linear with respect to addition and scalar multiplication in each
variable, the new term constructs $\prodterm{-,-}$, $\lproj{}$,
$\rproj{}$ are also linear in each variable. In particular, one can
check that $\prodterm{s+s',t+t'} =
\prodterm{s,t}+\prodterm{s',t}+\prodterm{s,t'}+\prodterm{s',t'}$.
These term constructs are introduced in the simply-typed
lambda-calculus of Section~\ref{sec:st}.

\subsection{Breaking consistency}
\label{sec:broken}

Although the set of requirements looks reasonable, as was shown in
\cite{arrighi08linear}, the equational system is not sound. Indeed,
given any term b one can construct the term
$
Y_b = \cocanon{(\lambda x.\canon{\cocanon{xx}+b})
(\lambda x.\canon{\cocanon{xx}+b})}
$
verifying the reduction
\begin{equation}
  \label{eq:yb}
  Y_b \to Y_b + b.
\end{equation}
This creates a problem of consistency, as enlightened in the following
sequence of equalities:
\begin{equation}
  \label{eq:broken}
  \zeroterm = Y_b-Y_b = (Y_b+b)-Y_b = b+(Y_b-Y_b) = b.
\end{equation}
This successfully shows that any term can be equated to
$\zeroterm$, rendering the system inconsistent.

\section{A simply-typed lambda-calculus}
\label{sec:st}

The problem occurring in Section~\ref{sec:broken} is due to the
possibility of constructing diverging terms.  In this section we study
a simply-typed, algebraic lambda-calculus. Equipped with a naive
reduction system, it verifies strong normalization. This allows us in
Section~\ref{sec:addingdiv} to analyze carefully the pitfalls occurring
when adding divergence.

\begin{table*}[tbh]
\begin{center}
\def\mysep{{\Rightarrow}}
\begin{tabular}{cc}
\hline
\hline
$\begin{array}{@{}l@{}lll}
  &
  &&
  \Delta,x:A\entail x:A,
  \\
  &
  &&
  \Delta\entail \puterm:\putype,
  \\
  &
  &&
  \Delta\entail \zeroterm:A
  \\
  &\Delta,x:A\entail s:B
  &\mysep&
  \Delta\entail\lambda x.s:A\duploli B,
\end{array}$
&
$\begin{array}{@{}l@{}lll}
  &\begin{array}{@{}ll@{}}
    \Delta\entail s:A\duploli B
    \\
    \Delta\entail t:A
  \end{array}\bigg\}
  &\mysep&
  \Delta\entail st:B,
\end{array}$
\\
$\begin{array}{@{}l@{}lll}
  &\Delta\entail s:A\prodtype B
  &\mysep&
  \Delta\entail \lprojp{s}:A,
  \\
  &\Delta\entail s:A\prodtype B
  &\mysep&
  \Delta\entail \rprojp{s}:B,
\end{array}$
&
$\begin{array}{@{}l@{}lll}
  &\begin{array}{@{}ll@{}}
    \Delta\entail s:A
    \\
    \Delta\entail t:B
  \end{array}\bigg\}
  &\mysep&
  \Delta\entail \prodterm{s,t}:A\prodtype B,
\end{array}$
\\
$\begin{array}{@{}l@{}lll}
  &\Delta\entail s:A
  &\mysep&
  \Delta\entail \alpha\cdot s:A,
  \\
  &\Delta\entail s:\monad{A}
  &\mysep&
  \Delta\entail \cocanon{s}:A,
  \\
  &\Delta\entail s:A
  &\mysep&
  \Delta\entail \canon{s}:\monad{A}.
\end{array}$
&
$\begin{array}{@{}l@{}lll}
  &\begin{array}{@{}ll@{}}
    \Delta\entail s:A
    \\
    \Delta\entail t:A
  \end{array}\bigg\}
  &\mysep&
  \Delta\entail s+t:A,
\end{array}$
\\
\hline
\end{tabular}
\end{center}
\caption{Typing rules.}
\label{tab:typrules}
\end{table*}

\begin{definition}
  \label{def:stlc}\rm
  We suppose the existence of a ring $\alg{A}$, containing a
  multiplication and an addition. A simply-typed, call-by-value,
  algebraic lambda-calculus called the \define{computational algebraic
  lambda-calculus} is constructed as follows. Types are of
  the form
  \begin{align*}
    A,B &\isdefby \iota\bor A\duploli B\bor A\prodtype B\bor\putype
    \bor\monad{A},
    \intertext{%
      where $\iota$ ranges over a set of type constants.
      Terms again come in two flavors:%
    }
    s,t &\isdefby x\bor \lambda x.s\bor st
    \bor \prodterm{s,t} \bor \lprojp{s}\bor\rprojp{s}\bor\puterm\bor
    s+t\bor \alpha\cdot s\bor\zeroterm
    \bor\canon{s}\bor\cocanon{s},
    \\
    u,v &\isdefby x\bor \lambda x.u\bor uv
    \bor \prodterm{u,v} \bor \lprojp{u}\bor\rprojp{u}\bor\puterm\bor
    \canon{s},
  \end{align*}
  where $\alpha\in\alg{A}$. Terms of the form $s,t$ are called
  \define{computations} and terms of the form $u,v$ are called
  \define{values}. The term $\canon{s}$ is the closure of a
  computation: such a term is not linear and can be duplicated ``as
  it''. The term construct $\cocanon{-}$ breaks such a closure and
  ``runs'' the computation.
  
  We define the notions of typing context $\Delta$ and of typing
  derivation $\Delta\entail s:A$ in the usual
  way~\cite{pierce02types}.
  Terms are considered
  up to $\alpha$-equivalence, and valid typing derivations are built
  using the rules of Table~\ref{tab:typrules}.
\end{definition}

\subsection{Small-step semantics}
\label{sec:opsemst}

The type system is valid with respect to the reduction system
described in Table~\ref{tab:betared}, modulo the addition of rules
for the added term constructs concerning the product. In the
following, we use the terminology of~\cite{arrighi08linear}.

\begin{definition}\rm
  Given any relation $R$ on terms, we say that it is a
  \define{call-by-value congruent relation} if for all pairs
  $(s,s'),(t,t')\in R$, the pairs $(st,st')$, $(st,s't)$,
  $(s+t,s+t')$, $(s+t,s'+t)$, $(\prodterm{s,t},\prodterm{s,t'})$,
  $(\prodterm{s,t},\prodterm{s',t})$, $(\rproj{s},\rproj{s'})$,
  $(\lproj{s},\lproj{s'})$, $(\alpha\cdot s,\alpha\cdot s')$ and
  $(\cocanon{s},\cocanon{s'})$ are in $R$.  We say that $R$ is
  \define{congruent} if it is call-by-value congruent and if for all
  pairs $(s,s')\in R$, we also have $(\lambda x.s,\lambda
  x.s')$, $(\canon{s},\canon{s'})$ in $R$.
\end{definition}

\begin{definition}\label{def:ac}\rm
  We define $\eqac$ to be the smallest congruent, equivalent relation
  on terms satisfying $s+t\eqac t+s$ and $r+(s+t)\eqac(r+s)+t$.
  We say that a relation $R$ is \define{consistent with $\eqac$} if
  $s\eqac s'Rt'\eqac t$ implies $sRt$.
\end{definition}

\begin{definition}\label{def:neutral}\rm
  A \define{normal} term $s$ is such that there does not exist
  a term $t$ with $s\to t$.  A \define{rewrite sequence} is a
  sequence $(s_i)_i$ of terms such that for all $i$, either $s_i\to
  s_{i+1}$ or $s_i$ is normal and $i$ is the last index of the
  sequence.
\end{definition}

\begin{definition}\label{def:betared}\rm
  We define the call-by-value reduction systems $E,F,A$ and $B$ of
  terms as the smallest call-by-value congruent relations consistent
  with $\eqac$, satisfying the rules in Table~\ref{tab:betared} where
  $B$ is augmented with the rules $\lproj{\prodterm{u,v}}\to u$ and
  $\rproj{\prodterm{u,v}}\to v$.
  In all the given rules,
  the terms $u,v$ are assumed to be values.  We write $L$ for the
  relation $A\cup B\cup E\cup F$.
\end{definition}

\begin{convention}\rm
  If $R$ is a relation, we write $s\to_R t$ in place of $(s,t)\in R$.
  We simply write $\to$ in place of $\to_L$, and if $s\to t$, we say
  that $s$ reduces to $t$. We denote with $\to^*_R$ the reflexive,
  transitive closure of $\to_R$.
\end{convention}

\begin{lemma}[Substitution]
  \label{lem:stsubst}
  Let $\Delta\entail v:A$ and $\Delta,x:A\entail s:B$ be two valid
  typing derivations, where $v$ is a value. Then $\Delta\entail
  s[x\leftarrow v]:B$ is a valid typing derivation.
\end{lemma}

\begin{proof}
  By structural induction on the typing derivation of
  $\Delta,x:A\entail s:B$.
\end{proof}

\begin{lemma}[Subject reduction]
  \label{lem:stsubred}
  Let $\Delta\entail s:A$ be a valid typing judgment such that $s\to
  t$. Then $\Delta\entail t:A$ is also valid.
\end{lemma}

\begin{proof}
  Proof by structural induction on the term $s$ and inspection of the
  reduction rules, using Lemma~\ref{lem:stsubst} for the first rule of
  group B.
\end{proof}

\begin{theorem}[Safety]
  Suppose that ${}\entail s:A$ is a valid typing judgment. Then
  either $s\to t$ with ${}\entail t:A$, or $s$ is 
  normal.
\end{theorem}

\begin{proof}
  By case distinction on the structure of $s$, using
  Lemma~\ref{lem:stsubred}.
\end{proof}

As for the simply-typed lambda-calculus, the reduction system is
normalizing. The proof uses the fact that the rewrite system consists
of two parts: the rules of groups E,F,A and the rules of group B.

\begin{lemma}
  \label{lem:efanorm}
  Let $s$ be any term. There exists an index $n_s$ such that any
  rewrite sequence $(s_i)_i$ in $E\cup F\cup A$ with $s_0=s$ consists
  of at most $n_s$ elements.
\end{lemma}

\begin{proof}
  We define two measures on terms. First, the ``plus-number of $s$'',
  written ${\it np}(s)$, and defined by ${\it np}(\zeroterm)={\it
    np}(x)={\it np}(\puterm)=1$, ${\it np}(\lambda x.s)={\it
    np}(\rprojp{s})={\it np}(\lprojp{s})={\it np}(\cocanon{s})={\it
    np}(\alpha\cdot s)=2\,{\it np}(s)$, ${\it np}(st)={\it
    np}(\prodterm{s,t})=2\,{\it np}(s)\,{\it np}(t)$, and ${\it
    np}(s+t)=1+{\it np}(s)+{\it np}(t)$. Then, the ``scalar-complexity
  of $s$'', written ${\it cx}(s)$, and defined by ${\it
    cx}(\zeroterm)={\it cx}(x)={\it cx}(\puterm)=1$, ${\it cx}(\lambda
  x.s)={\it cx}(\rprojp{s})={\it cx}(\lprojp{s})={\it
    cx}(\cocanon{s})=2\,{\it cx}(s)$, ${\it cx}(st)={\it
    cx}(\prodterm{s,t})={\it cx}(s+t)=2\,{\it np}(s)\,{\it np}(t)$,
  and ${\it cx}(\alpha\cdot s)=1+{\it cx}(s)$.  The lemma is proved by
  induction on $({\it np}(s),{\it cx}(s))$ with the
  lexicographic order.
\end{proof}

\begin{theorem}[Normalization]
  \label{the:stnorm}
  Let ${}\entail s:A$ be a valid typing judgment. There exists an index
  $n_s$ such that any rewrite sequence $(s_i)_i$ with $s_0=s$
  is finite and of at most $n_s$ elements.
\end{theorem}

\begin{proof}
  The proof uses reducibility candidates, and follows the proof
  provided in {\cite{girard89proofs}}. Lemma~\ref{lem:efanorm} is
  used to handle the cases where addition and
  scalar multiplication are involved.
\end{proof}

\begin{theorem}[Confluence]\label{the:conf}
  Suppose that $s$ is typable. If $s\to^*t$ and $s\to^*t'$,
  there exists a term $r$ such that $t\to^*r$
  and $t'\to^*r$.
\end{theorem}

\begin{proof}
  We first prove that for all terms $s$, if $s\to t$ and $s\to t'$
  then there exists a term $r$ such that $t\to^*r$ and $t'\to^*r$. We
  then prove the theorem using strong normalization, by induction on
  the length of the longest sequence of reductions.
\end{proof}

\subsubsection{Example: simulating quantum computation}

As an example of the expressiveness of the language, we follow the
motivation of \cite{arrighi08linear} and show that we
can simulate quantum computation using the computational algebraic
lambda-calculus.

Quantum computation is a paradigm where data is encoded on the state
of objects governed by the law of quantum physics. The mathematical
description of a quantum boolean is a (normalized) vector in a
$2$-dimensional Hilbert space $\mathbb{H}$. In order to give sense to
this vector, we choose an orthonormal basis $\s{\ket0,\ket1}$. A
vector $\alpha\ket0+\beta\ket1$ is understood as the ``quantum
superposition'' of the boolean $0$ and the boolean $1$. 

\def\qbit{\it qbool}
For simulating quantum computation, we therefore choose the ring
$\alg{A}$ to be the field of complex numbers. Given an arbitrary type
$X$, we can represent a quantum boolean in the computational
algebraic lambda-calculus as a closed value of type 
${\it qbool}=MX\to(MX\to
MX)$. We encode $\alpha\ket0+\beta\ket1$ as $\lambda
xy.\canon{\alpha\cdot\cocanon{x}+\beta\cdot\cocanon{y}}$. We write
$\trueterm$ for $\lambda xy.\canon{\cocanon{x}}$ and $\falseterm$ for
$\lambda xy.\canon{\cocanon{y}}$.

The operations we can perform on quantum booleans are of two sorts:
Quantum gates and measurements. In the mathematical description, the
former correspond to unitary maps. The Hadamard gate is such a
unitary, sending $\ket0$ to $\frac1{\sqrt2}(\ket0+\ket1)$ and $\ket1$
to $\frac1{\sqrt2}(\ket0-\ket1)$. It can be written as the term
\[
H =
\lambda x.
\lambda ab.
\canon{\s{x\canon{\begin{array}{@{}c@{}}
        \frac1{\sqrt2}\end{array}{\cdot}(\s{a}+\s{b})}
    \canon{\begin{array}{@{}c@{}}\frac1{\sqrt2}
      \end{array}{\cdot}(\s{a}-\s{b})}}}
\]
of type $\qbit\to\qbit$. 
Applying the Hadamard gate to a quantum boolean
$b$ is computing the term $Hb$.

A measurement has a probabilistic outcome and does not have a
satisfactory description as function of $\mathbb{H}$.
It is customary to
represent quantum booleans with \define{density matrices}, that is,
positive matrices of norm one. The measurement operation becomes the
map sending a matrix to its diagonal.

In order to model measurements, we can use the fact that the language
features higher-order terms and we encode a positive matrix as a term of
type $\qbit\to\qbit$. The quantum boolean $\alpha\ket0+\beta\ket1$ is
encoded as the term $v$ equal to
\[
\lambda x.
\lambda ab.
\canon{\s{
  x\,\canon{
    \alpha\bar{\alpha}{\cdot}\s{a}+\alpha\bar{\beta}{\cdot}\s{b}
  }\canon{
    \bar{\alpha}\beta{\cdot}\s{a}+\beta\bar{\beta}{\cdot}\s{b}
  }}
}.
\]
The application of the Hadamard gate to $v$ is $H'v$, where $H'$ is
the term $H'=\lambda x.H(xH)$ of type
$(\qbit\to\qbit)\to(\qbit\to\qbit)$.
The measurement is also of type $(\qbit\to\qbit)\to(\qbit\to\qbit)$ and
can be encoded as the term $P$ equal to
$
\lambda v.
\lambda x.
\lambda ab.
\canon{\s{
  (v\,x)\canon{\cocanon{a}}\canon{\zeroterm}
  + (v\,x)\canon{\zeroterm}\canon{\cocanon{b}}
}}.
$
We can check that $Pv$ is indeed equal to
$
\lambda x.
\lambda ab.
\canon{\s{
  x\,
  \canon{\alpha\bar{\alpha}{\cdot}\s{a}}
  \canon{\beta\bar{\beta}{\cdot}\s{b}
  }}
}.
$

\subsection{Equational theory}
\label{sec:steqth}

\begin{table*}[tbh]
\begin{center}
\begin{tabular}{c}
\hline
\hline
\mbox{\begin{minipage}{4.8in}
\mbox{%
\renewcommand{\arraystretch}{1.3}%
$\begin{array}{@{}c@{}}
\begin{array}{@{}r@{{~}\eqax{~}}l@{\quad}r@{{~}\eqax{~}}l@{}}
\hspace{-5ex}(*)\quad0\cdot s & \zeroterm
&
s+\zeroterm & s
\\
1\cdot s & s
&
\alpha\cdot s + \alpha\cdot t & \alpha\cdot(s+t)
\\
\alpha\cdot s+\beta\cdot s & (\alpha+\beta)\cdot s
&
(r+s)+t & r+(s+t)
\\
\alpha\cdot(\beta\cdot s) & (\alpha\beta)\cdot s
&
s+t & t+s
\\[2ex]
\prodterm{r+\alpha\cdot s,t}&\prodterm{r,t}+\alpha\cdot\prodterm{s,t}
&
\lprojp{s+\alpha\cdot t}&\lprojp{s}+\alpha\cdot\lprojp{t}
\\
\prodterm{r,s+\alpha\cdot t}&\prodterm{r,s}+\alpha\cdot\prodterm{r,t}
&
\rprojp{s+\alpha\cdot t}&\rprojp{s}+\alpha\cdot\rprojp{t}
\\
\prodterm{\zeroterm,t}&\zeroterm
&
\lprojp{\zeroterm}&\zeroterm
\\
\prodterm{t,\zeroterm}&\zeroterm
&
\rprojp{\zeroterm}&\zeroterm
\\
(r+\alpha\cdot s)t&rt+\alpha\cdot(st)
&
\zeroterm t&\zeroterm
\\
r(s+\alpha\cdot t)&rs+\alpha\cdot(rt)
&
t\zeroterm&\zeroterm
\\
\lambda x.(s+\alpha\cdot t)&\lambda x.s+\alpha\cdot(\lambda x.t)
&
\lambda x.\zeroterm&\zeroterm
\\
\cocanon{s+\alpha\cdot t}&\cocanon{s}+\alpha\cdot\cocanon{t}
&
\cocanon{\zeroterm}&\zeroterm
\\[2ex]
\lproj{\prodterm{u,v}}&u 
&
\canon{\cocanon{u}}&u
\\
\rproj{\prodterm{u,v}}&v
&
\cocanon{\canon{s}}&s
\\
\prodterm{\lprojp{u},\rprojp{u}}&u
&
(\lambda x.\cocanon{s})t & \cocanon{(\lambda x.s)t}
\\
(\lambda x.u)v & u[v/x]
&
((\lambda xy.r)s)t & ((\lambda yx.r)t)s
\\
\lambda x.(u x) & u
&
(\lambda x.r)((\lambda y.s)t)
& (\lambda y.(\lambda x.r)s)t
\\
(\lambda x.x)s&s
& u & \puterm
\end{array}
\\\\[-2ex]
\end{array}$}
\end{minipage}}
\\
\hline
\end{tabular}
\end{center}
\caption{Axiomatic equivalence relation.}
\label{tab:eqax}
\end{table*}

Together with its type system, the computational algebraic
lambda-calculus shares some strong similarities with Moggi's
computational lambda-calculus~\cite{moggi91notions} (although the
notations used for the monad term constructs are closer
to~\cite{filinski94representing}). We follow the same path for
defining a model for the algebraic lambda-calculus.

\begin{definition}\rm
  We define an equivalence relation $\eqax$ on terms as the smallest
  congruent equivalence relation consistent with $\eqac$, closed under
  $\alpha$-equivalence and the equations of Table~\ref{tab:eqax}. The
  relation is the symmetric closure of the reduction $L$ of
  Table~\ref{tab:betared}, together with the rules taking into account
  the new term constructs.
  
  Two valid typing judgments $\Delta\entail s,t:A$ are said to be
  \define{axiomatically equivalent}, written $\Delta\entail s\eqax
  t:A$, if $s\eqax t$ is provable.
\end{definition}

\begin{definition}\rm
  We define a $\alg{A}$-enriched computational category to be
 a cartesian closed category 
    $(\cat{C},\prodfunc,\internalhom,\punit)$,
 together with a strong monad $(M,\unit,\mult,\strength)$,
 such that the Kleisli category is enriched over
    the category of $\alg{A}$-modules.
  We refer the reader to the literature for the definitions
  (e.g. \cite{moggi91notions,kelly82basic,lambek89introduction}).
\end{definition}

\begin{example}\label{ex:set}
  The category of sets and functions together with the monad $M$
  sending a set $X$ to the free module generated by $X$ is a
  $\alg{A}$-enriched computational category.
\end{example}

\begin{definition}\rm
  We define the category $\linealcat$ as follows: objects are types
  and morphisms $A\to B$ are axiomatic equivalent classes of
  typing judgments $x:A\entail v:B$ (where $v$ is a value).
\end{definition}

\begin{theorem}
  \label{the:lincat}
  The category $\linealcat$ is a $\alg{A}$-enriched computational
  category. The cartesian closed structure is given by the classical
  subset of the language in the usual way (see e.g. 
  \cite{lambek89introduction}).
  The monad $M$ sends $A$ to $MA$ and $x:A\entail u:B$ to $y:MA\entail
  \canon{(\lambda x.u)\cocanon{y}}:MB$, and the three required
  morphisms are
  $\unit_A= x:A\entail\canon{x}:MA$,
  $\mult_A=
    x:MMA\entail\canon{\cocanon{\cocanon{x}}}:MA$,
    $\strength_{A,B}=
    x:MA\times B\entail \canon{\prodterm{\s{\lprojp{x}},\rprojp{x}}}
    :M(A\times B)$.
  The enrichment of $\linealcat(A,MB)$ is given by the module structure
  of the term algebra. Consider the two maps $f=(x:A\entail u:MB)$ and
  $g=(x:A\entail v:MB)$. We define $0 = (x:A\entail 
  \canon{\zeroterm}:MB)$,
$f+g=(x:A\entail \canon{\cocanon{u}+\cocanon{v}}:MB)$,
$\alpha\cdot f=(x:A\entail \canon{\alpha\cdot \cocanon{u}}:MB)$.
  \qed
\end{theorem}

\begin{definition}
  \label{def:linealdenot}\rm
  Consider a $\alg{A}$-enriched computational category $\cat{C}$. We
  define the interpretation of a computation 
  $\denot{\Delta\entail t:B}^c$ as a morphism in
  $\kcat{C}{M}$ and the interpretation of a value
  $\denot{\Delta\entail v:B}^v$ as a morphism in
  $\cat{C}$. They are defined inductively, together with their obvious
  meanings.
\end{definition}

\begin{theorem}
  \label{the:linealinternal}
  If we interpret the computational algebraic lambda-calculus in
  $\linealcat$ then the  equations  $\denot{x:A\entail
    v:B}^v\eqax(x:A\entail v:B)$ and $\denot{x:A\entail
    t:B}^c\eqax(x:A\entail\canon{t}:\monad{B})$ hold.\qed
\end{theorem}

\subsection{Relation with other algebraic lambda-calculi}

In this section, we relate the computational algebraic
lambda-calculus we
described in the previous section and the algebraic lambda-calculus
$\lambda_{\it alg}$ of Vaux \cite{vaux08algebraic} and lineal, the
algebraic lambda-calculus $\lambda_{\it lin}$ of Arrighi, Dowek and
D\`{i}az-Caro \cite{arrighi08linear,sts09}. Both languages can be
written using the term grammar
$s,t::= x\bor\lambda x.s\bor st\bor s+t\bor
\zeroterm\bor \alpha\cdot s$.
A possible simple type system is
$A,B::= \iota\bor A\duploli B,
$
where $\iota$ is a base type. The typing rules are the usual ones for
the application and the lambda-abstraction. For the sum, the zero and
the scalar multiplication, we use the typing rules found in
Table~\ref{tab:typrules}.

The main difference between the two languages is the reduction system.

\paragraph{Vaux's lambda-calculus.}
In $\lambda_{\it alg}$, the lambda-abstraction is linear: $\lambda
x.(s+t)\to \lambda x.s + \lambda x.t$, the application is linear on
the left and non-linear on the right: $(r+s)t\to rt+st$ but
$r(s+t)\not\to rs + rt$. However, $(\lambda x.s)t\to s[t/x]$ for any
term $t$.

\def\denota#1{{(\!|#1|\!)_{\it alg}}}
This language is call-by-name: a
function is fed with a computation (that is, a term in
superposition). One can encode $\lambda_{\it alg}$ in the
computational algebraic lambda-calculus as follows: $\denota{x} =
\cocanon{x}$, $\denota{\lambda x.s} = \lambda x.\denota{s}$,
$\denota{st} = \denota{s}\canon{\denota{t}}$. Types are encoded as
follows: $\denota{\iota} = \iota$, $\denota{A\duploli
  B}=M\denota{A}\duploli\denota{B}$. 

If $x:A\entail s:B$ is a valid typing judgment in $\lambda_{\it
  alg}$, $x:MA\entail \denota{s}:\denota{B}$ is valid in the
computational algebraic lambda-calculus. In particular, if
$\mathcal{C}$ is a $\mathcal{A}$-enriched computational model, $s$
described a map $M\denot{A}\to M\denot{B}$ in the category $\mathcal{C}$.

\paragraph{Lineal.} In $\lambda_{\it lin}$, the lambda-abstraction is
non-linear: $\lambda x.(s+t) \not\to \lambda x.s + \lambda x.t$. In
this calculus, the application is bilinear. In particular,
 $(\lambda x.s)u\to s[u/x]$ only if $u$ is a~value.

\def\denotl#1{{(\!|#1|\!)_{\it lin}}}
This calculus is call-by-value: the argument of a function is first
reduced to a value before being substituted in the body of the 
function.
One can encode $\lambda_{\it lin}$ in the
computational algebraic lambda-calculus as follows:
$\denotl{x} = x$,
$\denotl{\lambda x.s} = \lambda x.\canon{\denotl{s}}$,
$\denotl{st} = \cocanon{\denotl{s}\denotl{t}}$.
Types are encoded as
follows: $\denotl{\iota} = \iota$,
$\denotl{A\duploli B}=\denotl{A}\duploli M\denotl{B}$.

If $x:A\entail s:B$ is a valid typing judgment in $\lambda_{\it
  lin}$, $x:A\entail \denotl{s}:\denotl{B}$ is valid in the
computational algebraic lambda-calculus. In particular, if
$\mathcal{C}$ is a $\mathcal{A}$-enriched computational model, $s$
describes a morphism $\denot{A}\to M\denot{B}$ of $\mathcal{C}$.

\section{Adding controlled divergence}
\label{sec:addingdiv}

Because of Theorem~\ref{the:stnorm}, the term $Y_b$ of
Equation~\eqref{eq:yb} is not constructable in the computational
algebraic lambda-calculus. In this section, we add to the language a
notion of fixpoint in order to understand what goes wrong in the
untyped system.

\subsection{A fixpoint operator}
\label{sec:fixpoint}

In order to stay typed and to be able to keep most of the
computational interpretation of Section~\ref{sec:steqth} but still
to be able to have a term $Y_b$, we add to the language a unary
term operator $Y$ satisfying the reduction
$
Y(v)\to \cocanon{v\canon{Y(v)}}
$, 
linear with respect to the module structure
and satisfying the typing rule
\begin{equation}
  \label{eq:typrulyb}
  \Delta\entail s:MA\duploli MA
  \quad
  \Longrightarrow
  \quad
  \Delta\entail Y(s):A.
\end{equation}
We can now build a term $Y_b$ behaving as
required in Equation~\eqref{eq:yb}:
\begin{equation}
  \label{eq:styb}
  Y_b \equiv Y(\lambda x.\canon{b+\cocanon{x}}).
\end{equation}
Indeed, $Y(\lambda x.\canon{b+\cocanon{x}})$ reduces to the term
$\cocanon{(\lambda x.\canon{b+\cocanon{x}})
  \canon{Y(\lambda x.\canon{b+\cocanon{x}})}}$, which reduces to
$\cocanon{\canon{b+\cocanon{\canon{
	Y(\lambda x.\canon{b+\cocanon{x}})}}}}$,
itself reducing to $b$ $+$ 
$Y(\lambda x.\canon{b+\cocanon{x}})$.
Provided that $\Delta\entail b:B$, the typing
judgment $\Delta\entail Y_b:B$ is valid.
Of course, if we keep the operational semantics of
Section~\ref{sec:st}, the system becomes as inconsistent as with the
untyped calculus.

\subsection{The zero in the algebra of terms}
\label{sec:zeroterm}

To understand what goes wrong, consider the typing judgment
$
x:MA\entail x-x:MA.
$
With the equational system of Section~\ref{sec:steqth}, this typing
judgment is equivalent to $x:MA\entail \zeroterm:MA$. We claim that
this interpretation is correct as long as the term $x$ ``does not
contain any potential infinity''. With the additional construct
$Y$, we can replace $x$ with $\canon{Y_a}$ (where $Y_a$ is
constructed as in Equation~\eqref{eq:styb}) for some term $a$ of type
$A$. Consider the two terms
\begin{gather*}
  (\lambda y.\puterm)((\lambda x.(x-x))\canon{Y_a}),~~\neweqnber{eq:yb1}
  \qquad
  (\lambda y.\cocanon{y})
  ((\lambda x.(x-x))\canon{Y_a}).~~\neweqnber{eq:yb2}
\end{gather*}
Term~\eqref{eq:yb1} reduces to $(\lambda
y.\puterm)(0\cdot\canon{Y_a})$ and then to $0\cdot\puterm$.  It is
reasonable to think that this is equivalent to $\zeroterm$, thus
making $0\cdot\canon{Y_a}$ also equivalent to $\zeroterm$.
Term~\eqref{eq:yb2}, on the contrary, reduces to $Y_a-Y_a$, the flawed
term of Equation~\eqref{eq:broken}.

The problem does not show up when writing the equation
$\canon{Y_a}-\canon{Y_a}=0\cdot \canon{Y_a}$ but when one equates it
with $\zeroterm$. The term $0\cdot \canon{Y_a}$ is a ``weak zero''. It
makes a computation ``null'' as long as it does not diverge (and there
is always a diverging term of any inhabited type by using the
construction~\eqref{eq:styb}).  Therefore, despite the fact that
$\alg{A}$ is a ring, the set of terms of the form $\alpha\cdot s$ for
a fixed term $s$ is only a commutative monoid: addition does not admit
an inverse, it only has an identity element $0\cdot s$.
This is consistent with previous
studies~\cite{vaux08algebraic,selinger03order}.

\subsection{Recasting the equational theory}
\label{sec:recasteq}

With the addition of fixpoints, the equational theory given
in Section~\ref{sec:steqth} is not valid. In the discussion of the
previous section, we noted that the module of terms needs to be
weakened to a commutative monoid by removing the rule $0\cdot s\eqax
\zeroterm$. This is the only required modification,
and one can rewrite the whole theory without this rule.

In the following, we do not consider the language extended with the
fixpoint combinator; instead, we give a general theory for possible
divergence in the context of a simple type system.

\begin{definition}\rm
  A \define{weak $\alg{A}$-module} is a module over $\alg{A}$ where
  $\alg{A}$ is seen as a semiring. In particular, a weak
  $\alg{A}$-module is only a commutative monoid, and $v - v = 0\cdot v
  \neq 0$. Given a set $X$, the \define{free weak $\alg{A}$-module
    over $X$} is the structure consisting of all the finite sums
  $\sum_i \alpha_i\cdot x_i$, where $\alpha_i\in\alg{A}$ and $x_i\in
  X$.
\end{definition}

\begin{definition}\rm
  A \define{weak $\alg{A}$-enriched computational category}
  consists of a cartesian closed category
    $(\cat{C},\prodfunc,\internalhom,\punit)$,
   together with a strong monad $(M,\unit,\mult,\strength)$,
   such that the Kleisli category $\kcat{C}{M}$ is enriched over
    the category of weak $\alg{A}$-modules.
\end{definition}

\begin{remark}
  As we saw in Section~\ref{sec:zeroterm}, the two zero-functions
  $x:A\entail\zeroterm:A$ and $x:A\entail 0\cdot x:A$ behave
  differently in general. In a weak $\alg{A}$-enriched computational
  category, the former is interpreted as the unit element of the
  monoid $\kcat{C}{M}(A,A)$ whereas the latter is of the form
  $0\cdot{\it id}_A$, where ${\it id}_A$ is the identity map in
  $\kcat{C}{M}(A,B)$.
\end{remark}

\begin{lemma}
  Any $\alg{A}$-enriched computational category is also a weak
  $\alg{A}$-enriched computational category.
\end{lemma}

\begin{proof}
  Any $\alg{A}$-module is also a weak $\alg{A}$-module.
\end{proof}

\begin{remark}
  In particular, in a $\alg{A}$-enriched computational category, the
  two zero-functions $x:A\entail\zeroterm:B$ and $x:A\entail 0\cdot
  x:B$ are identified.
\end{remark}

\begin{definition}\rm
  Consider the typed language of Definition~\ref{def:stlc}, with the
  axiomatic equivalence of Table~\ref{tab:eqax} minus the very first
  rule, marked as $(*)$, stating $0\cdot u\eqax\zeroterm$. Let us call
  this language the \define{weak algebraic computational
    lambda-calculus} and the corresponding category of values
  $\wlinealcat$.
\end{definition}

\begin{theorem}
  1) The weak computational algebraic lambda-calculus is confluent.
  2) $\wlinealcat$ is a $\alg{A}$-enriched computational category.
  3) The weak computational algebraic lambda-calculus is an internal
  language for weak $\alg{A}$-enriched computational categories.
  \qed
\end{theorem}

\subsubsection{Extension of the language.}

Here, we assume that the language is extended to a call-by-value PCF
with a fixpoint combinator $Y$ and an algebraic structure, as follows
  \begin{align*}
    A,B &\isdefby \bit\bor\nat
    \bor A\duploli B\bor A\prodtype B\bor\putype
    \bor\monad{A},
    \\
    r,s,t &\isdefby x^A\bor \lambda x^A.s\bor st\bor 
    \prodterm{s,t} \bor \lprojp{s}\bor\rprojp{s}\bor\puterm\bor
    Y(s)\bor s+t\bor \alpha\cdot s\bor\zeroterm
    \bor\\
    &\phantom{{}\isdefby{}}\canon{s}\bor\cocanon{s}\bor
    \trueterm\bor\falseterm\bor
    \ifthenelse{r}{s}{t}\bor
    \natzero\bor\succterm(s)\bor\predterm(s)\bor
    \iszero(s),
  \end{align*}
    where $\alpha\in\alg{A}$. The meaning of the terms is the usual
    one for PCF\cite{plotkin77lcf}.
    The terms $\trueterm$ and $\falseterm$
    respectively stand for the boolean true and the boolean false; the
    term $\ifthenelse{r}{s}{t}$ is the test function on $r$; the term
    $\natzero$ stands for the natural number $0$; the term
    $\iszero(s)$ tests whether $s$ is null or not; $\predterm$ and
    $\succterm$ are respectively the predecessor and the successor
    function; finally $Y$ is the fixpoint combinator of
    Section~\ref{sec:fixpoint}.
    The notion of value is defined as in 
    Definition~\ref{def:untypedlang}.

  The rewrite system of Section~\ref{sec:opsemst} can be reformulated
  for the algebraic PCF. Again, apart from the rule $(*)$ of
  Table~\ref{tab:betared} which is not valid, all the other ones are
  correct. The reduction systems $E,F,A$ and $B$ of terms as the
  smallest congruent relations consistent with $\eqac$, satisfying the
  rules in Table~\ref{tab:betared} where $B$ is augmented with the
  rules $Y(v)\to \cocanon{v\canon{Y(v)}}$, $\succterm(\predterm(u))\to
  u$, $\iszero(\natzero)\to\trueterm$,
  $\iszero(\succterm(u))\to\falseterm$, $\lproj{\prodterm{u,v}}\to u$,
  $\rproj{\prodterm{u,v}}\to v$, $\ifthenelse{\trueterm}{s}{t}\to s$,
  $\ifthenelse{\falseterm}{s}{t}\to t$, In all the given rules, the
  terms $u,v$ are assumed to be values. We write $L'$ for the relation
  $A\cup B\cup E\cup F$, and as before we write $\to$ in place of
  $\to_{L'}$.

\begin{remark}
  Again, the rewrite system verifies subject
  reduction and progress. However, the system does not satisfy weak
  normalization. For example, the typing derivation $\entail Y\lambda
  x.\canon{\cocanon{x}}:A$ is valid, and the term $Y\lambda
  x.\canon{\cocanon{x}}$ reduces to itself.
\end{remark}

\begin{example}\label{ex:pow}
  An element of $M(\nat)$ can be regarded as the encoding of a 
  polynomial as follows. The function
\[
{\it Exp} = Y\lambda f.\canon{\lambda
  nx.\ifthenelse{\iszero(n)}{\cocanon{x}}{
    \cocanon{f}(\predterm(n))\,x}}
\]
of type $\nat\to(M\putype \to M\putype)$ takes
an integer $n$ and returns the map sending
$\canon{\alpha\cdot\puterm}$ to
$\canon{\alpha^{n}\cdot\puterm}$.
The map ${\it Pow}:M(\nat)\to(M\putype\to M\putype)$ defined as $\lambda
x.{\it Exp}\,{x}$ takes as input
$\canon{\sum_i\beta_i\cdot\overline{n}_i}$ and return the map sending
$\canon{\alpha\cdot\puterm}$ to
$\canon{(\sum_i\beta_i\,\alpha^{n_i})\cdot\puterm}$.
\end{example}

\subsubsection[Concrete models based on Set]{Concrete
  models based on $\Set$}
The category $\Set$ of sets and functions can be made into a weak
$\mathcal{A}$-enriched computational category. It is also possible to
model the PCF extension of the language:
$\denot{\putype}=\s{\puterm}$, the one-element set,
$\denot{\nat}=\mathbb{N}$, the set of natural numbers, and
$\denot{\bit}=\s{0,1}$, the two-elements sets. The denotation of the
product is the product in $\Set$ and the denotation of $A\duploli B$
is the set of $\Set$-function between $\denot{A}$ and $\denot{B}$.
The corresponding term constructs have their obvious meanings.
Provided that the ring $\mathcal{A}$ is endowed with a suitable notion
of limit (for example, taking $\mathcal{A}$ to be the reals with the
usual topology), we give two monads that can be used and an intuition
on their operational interpretation.

\paragraph{Strong convergence.} The monad $M_s$ defined as
$M_s(X)=\prodterm{X}_{\mathcal{A}}\cup\s{\bot}$, with
$\prodterm{X}_{\mathcal{A}}$ is the free weak $\alg{A}$-module
generated from $X$. We can define a fixpoint of $f:M_s(A)\to M_s(A)$
as $\lim_n f^n(\bot)$ if it exists, $\bot$ otherwise. We define
$\denot{Y(s)}$ as the fixpoint of~$\denot{s}$.
  
In this model, the morphism $\denot{x:A\entail \zeroterm:B}$ is the
constant function of value $0\in\prodterm{X}_{\mathcal{A}}$ and the
morphism $\denot{x:A\entail Y\lambda x.\canon{\cocanon{x}}:B}$ is the
constant function of value $\bot$. Moreover any non-converging
well-typed term $s$ have the same denotation $\bot$.

The set $M\mathbb{N}$ is $\bot$ together with all the finite linear
combinations $\sum_i\alpha_i\cdot n_i$.  The image of $M\mathbb{N}$ by
the operator $\denot{{\it Pow}}$ of Example~\ref{ex:pow} is a set of
functions
$\overline{p}:\mathcal{A}\cup\s{\bot}\to\mathcal{A}\cup\s{\bot}$
sending $\bot$ to $\bot$ and $\beta\in\mathcal{A}$ to $p(\beta)$. The
functions $p$ are either constant of value $\bot$ (when $f$ is the
image of $\bot$) or polynomials (when $f$ is the image of a linear
combination).

\paragraph{Weak convergence.} Define the semiring
$\mathcal{A}\cup\s{\omega}$ by extending the semiring $\mathcal{A}$
with a new element $\omega$. The sum and the multiplication are
extended as follows: $\alpha\omega=\omega$, $\alpha+\omega=\omega$.
We set $M_w(X)=({\mathcal{A}\cup\s{\omega}})^X$, the functions from
$X$ to $\mathcal{A}\cup\s{\omega}$. The fixpoint of $f:M_w(A)\to
M_w(A)$ is defined as the map sending $x\in X$ to $\lim_n f^n(0)(x)$
if it exists, $\omega$ otherwise. As previously, the denotation of
$Y(s)$ is the fixpoint of $\denot{s}$.
  
Here, $\denot{x:A\entail \zeroterm:B}$ and
$\denot{x:A\entail Y\lambda x.\canon{\cocanon{x}}:B}$ are the constant
functions of value $0\in\prodterm{X}_{\mathcal{A}}$. However, all
diverging terms do not have the same image. For example, the term
$Y\lambda x.\canon{\natzero + \succterm{\cocanon{x}}}$ of type $\nat$
corresponds to the element $f\in M_w(\mathbb{N})$ sending all
$n\in\mathbb{N}$ to $1\in\mathcal{A}$.

In this model, the image of $M(\mathbb{N})$ by ${\it Pow}$ is the set of
(generalized) entire functions $\mathcal{A}\to\mathcal{A}$, sending
$\beta$ to $\sum_i\alpha_i(\beta)^{n_i}$. By ``generalized'', we mean
that the functions may send some $\beta$ to $\omega$.

\section{Conclusion}

In this paper, we sketched the required structures for a semantics for
a typed algebraic lambda-calculus and discussed relation with previous
works. We showed that the problems occurring with divergence can be
solved by using a weak module. Finally, we described an algebraic PCF
and its interpretation in two concrete $\Set$-based models.

This raises the question of the complete description of the possible
operational behaviors of the algebraic PCF and the study of their
denotational semantics.

\section{Acknowledgments}

I would like to thank Gilles Dowek for introducing me to algebraic
calculi. I would also like to thank Pablo Arrighi and the research
group CAPP in Grenoble for helpful discussions.


\begin{thebibliography}{10}
\providecommand{\bibitemstart}[1]{\bibitem{#1}}
\providecommand{\bibitemend}{}
\providecommand{\bibliographystart}{}
\providecommand{\bibliographyend}{}
\providecommand{\url}[1]{\texttt{#1}}
\providecommand{\urlprefix}{Available at }
\providecommand{\bibinfo}[2]{#2}
\bibliographystart

\bibitemstart{sts09}
\bibinfo{author}{Pablo Arrighi} \& \bibinfo{author}{Alejandro D{\'i}az-Caro}
  (\bibinfo{year}{2009}).
\newblock \emph{\bibinfo{title}{A System {F} accounting for scalars}}.
\newblock \bibinfo{howpublished}{Preprint: arXiv:0903.3741}.
\bibitemend

\bibitemstart{arrighi08linear}
\bibinfo{author}{Pablo Arrighi} \& \bibinfo{author}{Gilles Dowek}
  (\bibinfo{year}{2008}): \emph{\bibinfo{title}{Linear-algebraic
  lambda-calculus: higher-order, encodings, and confluence.}}
\newblock In: {\sl \bibinfo{booktitle}{Proceedings of the 19th international
  conference on Rewriting Techniques and Applications (RTA'08)}}, {\sl
  \bibinfo{series}{Lecture Notes in Computer Science}} \bibinfo{volume}{5117},
  pp. \bibinfo{pages}{17--31}.
\bibitemend

\bibitemstart{barbanera93combining}
\bibinfo{author}{Franco Barbanera} \& \bibinfo{author}{Maribel Fern\'{a}ndez}
  (\bibinfo{year}{1993}): \emph{\bibinfo{title}{Combining first and
  higher-order rewrite systems with type assignment systems}}.
\newblock In: {\sl \bibinfo{booktitle}{Proceedings of the International
  Conference on Typed Lambda Calculi and Applications, TLCA'93}}, {\sl
  \bibinfo{series}{Lecture Notes in Computer Science}} \bibinfo{volume}{664},
  pp. \bibinfo{pages}{60--74}.
\bibitemend

\bibitemstart{barendregt84lambda}
\bibinfo{author}{Henk~P. Barendregt} (\bibinfo{year}{1984}):
  \emph{\bibinfo{title}{The Lambda-Calculus, its Syntax and Semantics}}.
\newblock \bibinfo{publisher}{North Holland}.
\bibitemend

\bibitemstart{blanqui99calculus}
\bibinfo{author}{Fr\'{e}d\'{e}ric Blanqui}, \bibinfo{author}{Jean-Pierre
  Jouannaud} \& \bibinfo{author}{Mitsuhiro Okada} (\bibinfo{year}{1999}):
  \emph{\bibinfo{title}{The Calculus of algebraic Constructions}}.
\newblock In: {\sl \bibinfo{booktitle}{RtA '99: Proceedings of the 10th
  International Conference on Rewriting Techniques and Applications}},
  \bibinfo{publisher}{Springer-Verlag}, \bibinfo{address}{London, UK}, pp.
  \bibinfo{pages}{301--316}.
\bibitemend

\bibitemstart{breazu91polymorphic}
\bibinfo{author}{Val Breazu-Tannen} \& \bibinfo{author}{Jean Gallier}
  (\bibinfo{year}{1991}): \emph{\bibinfo{title}{Polymorphic rewriting conserves
  algebraic strong normalization}}.
\newblock {\sl \bibinfo{journal}{Theoretical Computer Science}}
  \bibinfo{volume}{83}(\bibinfo{number}{1}), pp. \bibinfo{pages}{3--28}.
\bibitemend

\bibitemstart{ehrhard03differential}
\bibinfo{author}{Thomas Ehrhard} \& \bibinfo{author}{Laurent Regnier}
  (\bibinfo{year}{2003}): \emph{\bibinfo{title}{The differential
  lambda-calculus}}.
\newblock {\sl \bibinfo{journal}{Theoretical Computer Science}}
  \bibinfo{volume}{309}(\bibinfo{number}{1--2}), pp. \bibinfo{pages}{1--41}.
\bibitemend

\bibitemstart{filinski94representing}
\bibinfo{author}{Andrzej Filinski} (\bibinfo{year}{1996}):
  \emph{\bibinfo{title}{Representing Monads}}.
\newblock In: {\sl \bibinfo{booktitle}{Proceedings of the 21st ACM
  SIGPLAN-SIGACT Symposium on Principles of Programming Languages}}, pp.
  \bibinfo{pages}{446--457}.
\bibitemend

\bibitemstart{girard89proofs}
\bibinfo{author}{Jean-Yves Girard}, \bibinfo{author}{Yves Lafont} \&
  \bibinfo{author}{Paul Taylor} (\bibinfo{year}{1990}):
  \emph{\bibinfo{title}{Proofs and Types}}.
\newblock \bibinfo{publisher}{Cambridge University Press}.
\bibitemend

\bibitemstart{kelly82basic}
\bibinfo{author}{Gregory~M. Kelly} (\bibinfo{year}{1982}):
  \emph{\bibinfo{title}{Basic Concepts of Enriched Category Theory}}, {\sl
  \bibinfo{series}{London Mathematical Society Lecture Notes
  Series}}~\bibinfo{volume}{64}.
\newblock \bibinfo{publisher}{Cambridge University Press}.
\newblock \bibinfo{note}{Avalaible in Reprint in Theory and Application of
  Categories, No 10, 1982.}
\bibitemend

\bibitemstart{lambek89introduction}
\bibinfo{author}{Joachim Lambek} \& \bibinfo{author}{Philip Scott}
  (\bibinfo{year}{1989}): \emph{\bibinfo{title}{Introduction to Higher Order
  Categorical Logic}}.
\newblock \bibinfo{publisher}{Cambridge University Press}.
\bibitemend

\bibitemstart{moggi91notions}
\bibinfo{author}{Eugenio Moggi} (\bibinfo{year}{1991}):
  \emph{\bibinfo{title}{Notions of Computation and Monads}}.
\newblock {\sl \bibinfo{journal}{Information and Computation}}
  \bibinfo{volume}{93}, pp. \bibinfo{pages}{55--92}.
\bibitemend

\bibitemstart{pierce02types}
\bibinfo{author}{Benjamin~C. Pierce} (\bibinfo{year}{2002}):
  \emph{\bibinfo{title}{Types and Programming Languages}}.
\newblock \bibinfo{publisher}{MIT Press}.
\bibitemend

\bibitemstart{plotkin77lcf}
\bibinfo{author}{Gordon~D. Plotkin} (\bibinfo{year}{1977}):
  \emph{\bibinfo{title}{{LCF} Considered as a Programming Language}}.
\newblock {\sl \bibinfo{journal}{Theoretical Computer Science}}
  \bibinfo{volume}{5}, pp. \bibinfo{pages}{223--255}.
\bibitemend

\bibitemstart{selinger03order}
\bibinfo{author}{Peter Selinger} (\bibinfo{year}{2003}):
  \emph{\bibinfo{title}{Order-Incompleteness and Finite Lambda-Reduction
  Models}}.
\newblock {\sl \bibinfo{journal}{Theoretical Computer Science}}
  \bibinfo{volume}{309}, pp. \bibinfo{pages}{43--63}.
\bibitemend

\bibitemstart{tonder04lambda}
\bibinfo{author}{Andr{\'e} van Tonder} (\bibinfo{year}{2004}):
  \emph{\bibinfo{title}{A Lambda Calculus for Quantum Computation}}.
\newblock {\sl \bibinfo{journal}{SIAM Journal of Computing}}
  \bibinfo{volume}{33}, pp. \bibinfo{pages}{1109--1135}.
\bibitemend

\bibitemstart{vaux08algebraic}
\bibinfo{author}{Lionel Vaux} (\bibinfo{year}{2008}):
  \emph{\bibinfo{title}{Algebraic lambda-calculus}}.
\newblock {\sl \bibinfo{journal}{Mathematical Structures in Computer Science}}
  \bibinfo{note}{To appear.}
\bibitemend

\bibliographyend
\end{thebibliography}
\end{document}